\documentclass[12pt]{article}
\usepackage{amsmath,mathrsfs, eufrak, mathtools,amsthm}
\usepackage{amssymb}

\usepackage{bbm}
\usepackage{nicefrac}
\usepackage{newtxtext}
\usepackage{graphics,framed, wrapfig}
\usepackage{pgf}
\usepackage{tikz,pgfplots}
\usepackage{mhchem}
\usepackage{mathtools}

\usepackage{microtype}
\usepackage{titlesec}
\usepackage[top=1.0in,bottom=1.0in,left=1.5in,right=1in]{geometry}

\titleformat{\subsubsection}[runin]{\normalfont\bfseries}{\thesubsubsection}{1em}{}

\usepackage{pstricks,pst-plot,psfrag}

\usepackage{graphicx,subfigure,xspace,bm,tikz,pgfplots}
\usepackage[lined,linesnumbered,algoruled,noend]{algorithm2e}

\newtheorem{theorem}{Theorem}[section]
\newtheorem{definition}{Definition}

\newtheorem{lemma}{Lemma}

\newtheorem{remark}{Remark}

\newtheorem{example}{Example}

\newtheorem{proposition}[theorem]{Proposition}  
\renewcommand{\qed}{\hfill $\blacksquare$}

\newcounter{para}[section]

\newcommand{\oprocendsymbol}{\hbox{$\bullet$}}
\newcommand{\oprocend}{\relax\ifmmode\else\unskip\hfill\fi\oprocendsymbol}

\newcommand{\g}{\mathfrak{g}}

\newcommand{\BP}{P}
\newcommand{\BH}{H}

\newcommand{\cS}{\mathcal{S}}

\newcommand{\R}{\mathbb{R}}

\renewcommand{\sl}{\mathfrak{sl}}
\newcommand{\so}{\mathfrak{so}}

\newcommand{\C}{\mathbb{C}}

\newcommand*\Laplace{\mathop{}\!\mathbin\bigtriangleup}

\definecolor{BBlue}{cmyk}{.98,0.10,0,.25}

\begin{document}

\title{Ensemble Observability of Bloch Equations \\ with Unknown Population Density}
\date{}
\maketitle

\vspace{-2cm}
\begin{flushright}
{\small {Xudong Chen\footnote[1]{X. Chen is with the ECEE Dept., CU Boulder. Email: \texttt{xudong.chen@colorado.edu}.} 
}}
\end{flushright}

\begin{abstract}
We introduce in the paper a novel observability problem for a large population (in the limit, a continuum ensemble) of nonholonomic control systems with unknown population density. We address the problem by focussing on a prototype of such ensemble system, namely, the ensemble of Bloch equations which is known for its use of describing the evolution of the bulk magnetization of a collective of non-interacting nuclear spins in a static field modulated by a radio frequency ({\em rf}) field. The dynamics of the equations are structurally identical, but show variations in Larmor dispersion and {\em rf} inhomogeneity. We assume that the initial state of any individual system (i.e., individual Bloch equation) is unknown and, moreover, the population density of these individual systems is also unknown. Furthermore, we assume that at any time, there is only one scalar measurement output at our disposal. The measurement output integrates a certain observation function, common to all individual systems, over the continuum ensemble. The observability problem we pose in the paper is thus the following:  Whether one is able to use the common control input (i.e., the {\em rf} field) and the single measurement output to estimate both the initial states of the individual systems and the population density?  Amongst other things, we establish a sufficient condition for the ensemble system to be observable:  We show that if the common observation function is any harmonic homogeneous polynomial of positive degree, then the ensemble system is observable. The main focus of the paper is to demonstrate how to leverage tools from representation theory of Lie algebras to tackle the observability problem. Although the results we establish in the paper are for the specific ensemble of Bloch equations, the approach we develop along the analysis can be generalized to investigate observability of other general ensembles of nonholonomic control systems with a single,  integrated measurement output.     
\end{abstract}


\section{Introduction  and Main result}\label{sec:introduction}
We consider in the paper a large population (in the limit, a continuum) of independent control systems---these individual systems are structurally identical, but show variations in system parameters. We call such a population of control systems an {\em ensemble system}. A precise description of the system model will be given shortly. Control of an ensemble system is about broadcasting a finite-dimensional control input to simultaneously steer all the individual systems in the continuum ensemble. Questions such as whether an ensemble system is controllable and how to generate a control input to steer the entire population of systems have all been investigated to some extent in the literature. For control of linear ensembles (i.e., ensembles of linear systems), we refer the reader to~\cite{li2015ensemble,li2011ensemble,helmke2014uniform,chen2020controllability} and~\cite[Ch.~12]{fuhrmann2015mathematics}.
For control of nonlinear ensembles,  
we first refer the reader to the work~\cite{li2006control,li2009ensemble} by Li and Khaneja. The authors established controllability of a continuum ensemble of Bloch equations~\cite{bloch1946nuclear} using a Lie algebraic method. A similar controllability problem has also been addressed in~\cite{beauchard2010controllability}. But, the authors there have used a different approach that leverages tools from functional analysis. Continuum ensembles of bilinear systems for formation control has been investigated in~\cite{chen2019controllability}.  
We next refer the reader to~\cite{agrachev2016ensemble}  in which  the Rachevsky-Chow theorem (also known as the Lie algebraic rank condition) has been generalized so that it can be used a sufficient condition to check whether a continuum ensemble of control-affine systems is controllable. 
We have recently proposed in~\cite{chen2019structure} a novel class of ensembles of control-affine systems, termed distinguished ensembles, and shown that any such ensemble system satisfies the generalized version of the Rachevsky-Chow theorem and, hence, is ensemble controllable.    

 We address in the paper the counterpart of the ensemble control problem, namely the ensemble estimation problem. Roughly speaking, estimation of an ensemble system is about  
 using a single, integrated measurement output (of finite-dimension) to estimate the initial state of every individual system in the ensemble.     
Note that in its basic setup, the ensemble estimation problem is addressed under the assumption that the entire knowledge of the system model is available (See, for example,~\cite{chen2019structure}).  
We consider in the paper a more challenging but realistic scenario: {\em We assume that the underlying population density of the individual systems in the (continuum) ensemble is unknown}. 

The observability problem we will address in the paper is thus the problem about feasibility of estimating both the initial states and the population density of the individual systems in the ensemble. 
Note, in particular, that the problem can be viewed as a combination of two interrelated subproblems: One is the ``usual ensemble observability problem'' in which one has the complete knowledge of the ensemble model and aims to estimate the initial states of its individual systems. The other one can be related to the problem of ``system identification'' for which one treats the population density as an intrinsic parameter of an ensemble system.  

To the best of author's knowledge, the ensemble observability problem we posed here has not yet been addressed in the literature.  
One of the main contributions of the paper is thus to develop methods for tackling such a problem. Our methods rely on the use of representation theory of Lie algebras.  To demonstrate such a connection between the observability problem and the tools from the representation theory, we focus in the paper on a prototype of an ensemble of nonholonomic control systems, namely, a continuum ensemble of Bloch equations (the mathematical model will be given shortly). 
 Although the results established in the paper are for the specific ensemble of Bloch equations,  
 the methods we develop along the analysis can be extended to address other generals cases. We will address such an extension toward the end of the paper.

\subsection{System model: Ensemble of Bloch equations}
Bloch equation~\cite{bloch1946nuclear} is known for its use of describing the evolution of the bulk magnetization of a collective of non-interacting nuclear spins in a static field modulated by a controlled radio frequency ({\em rf}) field. 
When factors such as Larmor dispersion and {\em rf} inhomogeneity matter, a continuum ensemble of Bloch equations is often used to model variations in these system parameters.    
To this end, we let $S^2$ be the unit sphere embedded in $\R^3$. For a point $x\in S^2$, we let  $x = (x_1,x_2,x_3)$ be its coordinates. Next, we define three vector fields on $S^2$ as follows:
  \begin{equation}\label{eq:vectorfields}
  f_0(x):= 
  \begin{bmatrix}
  	x_2 \\
  	-x_1 \\
  	0
  \end{bmatrix}, \quad
  f_1(x):= 
  \begin{bmatrix}
  	x_3 \\
  	0 \\
  	-x_1 
  \end{bmatrix}, \quad
  f_2(x):=
  \begin{bmatrix}
  	0 \\
  	x_3 \\
  	-x_2
  \end{bmatrix}.
  \end{equation}  
  Then, the dynamics of an ensemble of Bloch equations, parametrized by a pair of scalar parameters $(\sigma_1,\sigma_2)$, are described by the following differential equations: 
  \begin{equation}\label{eq:justdynamics}
  \dot x_\sigma(t) = \sigma_1 f_0(x_\sigma(t)) + \sigma_2\sum^2_{i = 1}u_{i}(t) f_i(x_\sigma(t)),  
  \end{equation}
  where $u_1(t)$, $u_2(t)$ are scalar control inputs and the two parameters $\sigma_1$, $\sigma_2$ are used to model Larmor dispersion and {\em rf} inhomogeneity, respectively.  
 We assume in the paper that $\sigma_1\in [a_1,b_1]$ with $a_1 < b_1$ and $\sigma_2 \in [a_2, b_2]$ with $0< a_2 < b_2$.  
  We let $\sigma := (\sigma_1,\sigma_2)$ and $$\Sigma: = [a_1,b_1]\times [a_2, b_2].$$ 
  We call $\Sigma$ the {\em parameterization space}. 
  
  If an individual Bloch equation is associated with the parameter~$\sigma$, we call it {\em system-$\sigma$}. 
 Note that by~\eqref{eq:justdynamics}, each system-$\sigma$ is control-affine.  We call $f_0$ a {\em drifting vector field} and $f_1$, $f_2$ {\em control vector fields}.  
We note here that the same model~\eqref{eq:justdynamics} has been used in~\cite{li2006control,li2009ensemble,beauchard2010controllability} for the study of ensemble controllability problem.

For ease of notation, we let $u(t):= (u_1(t), u_2(t))$. Further, we let $x_\Sigma(t)$ be the collection of current states $x_\sigma(t)$ of all individual systems in the ensemble: 
$$x_\Sigma(t):= \{x_\sigma(t) \mid \sigma \in \Sigma\}.$$  
We call $x_\Sigma(t)$ a {\em profile}. Note that each profile $x_\Sigma(t)$ can be thought as a function from $\Sigma$ to~$S^2$. Let ${\rm C}^0(\Sigma, S^2)$ be the set of continuous functions from $\Sigma$ to $S^2$. We assume in the paper that each profile $x_\Sigma(t)$ belongs to ${\rm C}^0(\Sigma, S^2)$. 

  Next, we let $\mu$ be a positive Borel measure defined on the parameterization space~$\Sigma$. The measure $\mu$ will be used to describe the population density of the individual systems. Specifically, we assume that for any given measurable subset~$\Sigma'$ of~$\Sigma$, the total amount of individual systems, with their indices~$\sigma$ belonging to~$\Sigma'$, is proportional to $\int_{\Sigma'}d\mu$. For ease of analysis, we assume that there is a continuous function $\rho$ on~$\Sigma$ such that $\rho(\sigma)\ge 0$ for all $\sigma$ and  $d\mu = \rho(\sigma)d\sigma$. We call $\rho$ the {\em density function}.

 With the measure~$\mu$ defined above, we now introduce the estimation model  as a counterpart of~\eqref{eq:justdynamics}. Following the problem formulation in~\cite{chen2019structure}, we assume that there is only {\em one scalar} measurement output, denoted by $y(t)$, at our disposal. The measurement output integrates a certain observation function $\phi$ (common to all individual systems) over the entire parameterization space~$\Sigma$. Specifically, we have that
$$
		y(t) := \displaystyle\int_\Sigma \phi(x_\sigma(t)) d\mu,  
$$ 
where the observation function $\phi:S^2\to \R$ is assumed to be continuous. Consider, for example, the map $\phi:x\mapsto x_i$  for some $i = 1,2,3$. Then, $y(t)$ can be interpreted as the projection of the bulk magnetization vector to the $x_i$-axis. We consider in the paper general observation functions that can render the ensemble system observable. A precise problem formulation will be given soon.

By combining the control model~\eqref{eq:justdynamics} and the above estimation model, we obtain the following ensemble system:  
 \begin{equation}\label{eq:system}
\left\{
\begin{array}{lll}
	\dot x_\sigma(t) & = & \sigma_1 f_0(x_\sigma(t)) + \sigma_2\sum^2_{i = 1}u_{i}(t) f_i(x_\sigma(t)),  \vspace{3pt}\\
	y(t) & = & \displaystyle\int_\Sigma \phi(x_\sigma(t)) d\mu. 	
\end{array}
\right.
\end{equation}
We assume in the paper that  $x_\sigma(0)$ is unknown for all $\sigma\in \Sigma$ and, moreover, the measure~$\mu$ is also unknown. We note here that system~\eqref{eq:system} can be viewed as a prototype of a general ensemble of nonholonomic control systems with a single integrated measurement output.

\subsection{Problem formulation: Ensemble observability} 
We formulate in the section the ensemble observability problem for system~\eqref{eq:system} with unknown population density. We start with the following definition:    

\begin{definition}\label{def:outputequiv}
Let $x_\Sigma(0)$, $x'_\Sigma(0)$ be initial profiles and $\mu$, $\mu'$ be positive Borel measures on $\Sigma$.  
	Two pairs $(x_\Sigma(0), \mu) $ and $(x'_\Sigma(0), \mu')$ are {\bf output equivalent}, which we denote by $$(x_\Sigma(0), \mu) \sim (x'_\Sigma(0),  \mu'),$$ if for any $T > 0$ and any integrable function $u:[0,T]\to \R^2$ as a control input,  
	$$\int_\Sigma \phi(x_\sigma(t)) d\mu = \int_\Sigma \phi(x'_\sigma(t)) d\mu', \quad \forall t\in [0,T].$$ 	   
\end{definition}

Following the above definition, we introduce for each pair $(x_\Sigma(0), \mu)$, the collection of its output equivalent pairs as follows:
$$
O(x_\Sigma(0), \mu):= \{ (x'_\Sigma(0), \mu')  \mid  (x'_\Sigma(0), \mu') \sim (x_\Sigma(0), \mu) \}.
$$
Note that $(x_\Sigma(0), \mu)$ always belongs to $O(x_\Sigma(0), \mu)$. We next have the following definition: 

\begin{definition}\label{def:ensembleobservability}
System~\eqref{eq:system} is {\bf weakly ensemble observable} if for any given  $(x_\Sigma(0), \mu)$, the set $O(x_\Sigma(0), \mu)$ is finite. Moreover, we require that if $(x'_\Sigma(0), \mu')$ belongs to $O(x_\Sigma(0), \mu)$ and if $(x'_\Sigma(0), \mu') \neq (x_\Sigma(0), \mu)$, then the following hold:  
\begin{enumerate}
\item[(1)] The two measures $\mu'$ and $\mu$ are identical. 
\item[(2)] For any $\sigma\in \Sigma$, $x'_\sigma(0) \neq x_\sigma(0)$.   
\end{enumerate}
System~\eqref{eq:system} is {\bf ensemble observable} if for any $(x_\Sigma(0), \mu)$, $O(x_\Sigma(0), \mu) = \{(x_\Sigma(0), \mu)\}$. 
\end{definition}

\begin{remark}{\em 
We note that the above definition about (weak) ensemble observability is stronger than the ``usual'' definition of ensemble observability introduced  in~\cite{chen2019structure}. The key difference between the two definitions is that Def.~\ref{def:ensembleobservability} takes into account the fact that one needs to identify the unknown population density  as well. 
We also note that the two items in Def.~\ref{def:ensembleobservability} have the following implication: If system~\eqref{eq:system} is weakly ensemble observable, then by knowing the initial state $x_\sigma(0)$ of a single individual system-$\sigma$, one is able to estimate the entire initial profile $x_\Sigma(0)$ and the measure~$\mu$. }
\end{remark}

The problem we will address in the paper is the following: {\em Given the control dynamics~\eqref{eq:justdynamics}, what kind of observation function will guarantee that the entire system~\eqref{eq:system} is (weakly) ensemble observable?} 
 We provide below a partial solution to the above question by providing a class of observation functions that can fulfill the requirement.

 \begin{figure}[h]
\begin{center}
	\includegraphics[width = 0.6\textwidth]{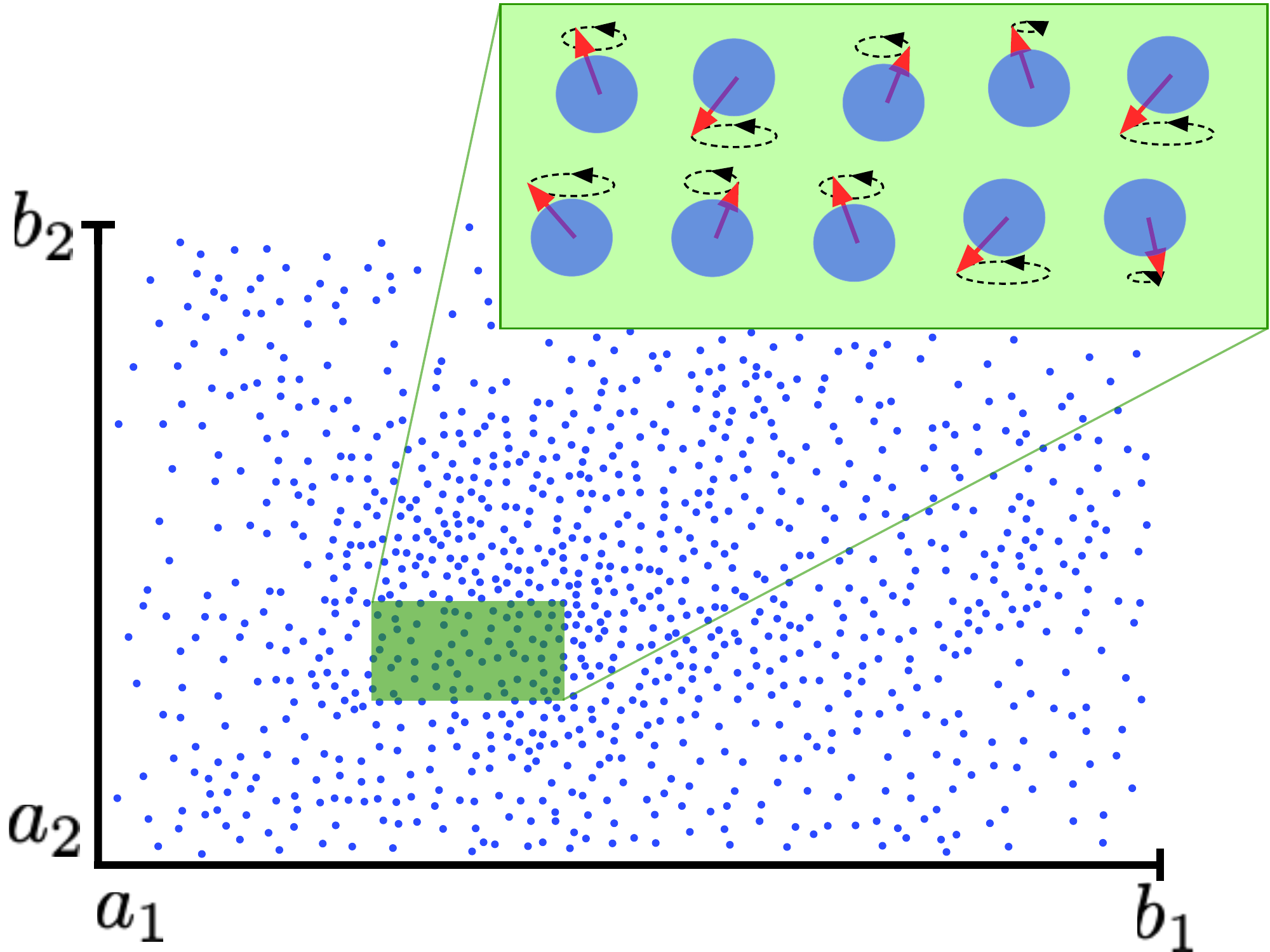}
	\caption{A large population (in the limit, a continuum) of Bloch equations over $\Sigma = [a_1,b_1]\times [a_2,b_2]$. Both initial states and the population density are unknown to the controller.}\label{fig:blocheq} 
	\end{center}
\end{figure}

\subsection{Main result}
We state in the subsection the main result of the paper.  
To proceed, we first introduce a few notations that are necessary to state the result.  
Let~$\BP$ be the space of all homogeneous polynomials in variables $x_1$, $x_2$, and $x_3$. For any nonnegative integer~$n$, we let $\BP_n$ be the space of homogeneous polynomials of degree~$n$. The dimension of $\BP_n$ is given by $\nicefrac{(n+2)(n+1)}{2}$. Note that one can treat a polynomial $p(x)$ as a function on $S^2$ by restricting~$x\in \R^3$ to~$x\in S^2$.  
Denote by $\Laplace$ the Laplace operator on $\R^3$: 
$$\Laplace:= \sum^3_{i = 1}\nicefrac{\partial^2}{\partial x^2_i}.$$ 
We recall the following definition: 

\begin{definition}
	A polynomial~$p$ is {\bf harmonic} if $\Laplace p = 0$. 
\end{definition}

Let $\BH_n$ be the space of harmonic homogeneous polynomials of degree~$n$. The dimension of $\BH_n$ is $2n + 1$. 
For example, for $n = 1$, $H_1$ is spanned by the basis $\{x_1, x_2, x_3\}$; for $n = 2$, $H_2$ is spanned by the basis $\{x^2_1 - x^2_2, x^2_2 - x^2_3,  x_1x_2, x_1x_3, x_2x_3\}$.  

For any real number $r$, we let $\lfloor r \rfloor$ be the largest integer such that $\lfloor r \rfloor \le r$. Then, it is known (see, for example,~\cite[Ch.~17.6]{hall2013quantum}) that the space of $P_n$ can be decomposed as a direct sum as follows: 
$$
\BP_n = \bigoplus^{\lfloor \nicefrac{n}{2} \rfloor}_{k = 0} \|x\|^{2k} \BH_{n-2k},
$$  
where $\|x\|^2 := \sum^3_{i = 1} x^2_i$. 

We will now state the main result of the paper:   
 
\begin{theorem}\label{thm:mthm}
 Consider the ensemble system~\eqref{eq:system}. Suppose that the observation function $\phi$ is nonzero and belongs to $H_n$ for some $n \ge 1$; then, the following hold:  
 \begin{enumerate}
 \item[(1)] If $n$ is even, then system~\eqref{eq:system} is weakly ensemble observable. Moreover, for any pair $(x_\Sigma(0), \mu)$, we have that  
 \begin{equation}\label{eq:ambiguity}
 O(x_\Sigma(0), \mu) = \{(x_\Sigma(0), \mu), (-x_\Sigma(0), \mu)\}.
 \end{equation}
 \item[(2)] If $n$ is odd, then system~\eqref{eq:system} is ensemble observable.
 \end{enumerate}
\end{theorem}

\begin{remark}{\em 
We note here that if $n$ is even, then $O(x_\Sigma,\mu)$ contains at least the two pairs in~\eqref{eq:ambiguity}.  
We elaborate below on the fact. First, note that if two initial profiles are related by $x'_\Sigma(0) = -x_\Sigma(0)$, then for any control input~$u(t)$, it always holds that $x'_\Sigma(t) = - x_\Sigma(t)$ for all~$t$. Next, note that if $\phi$ is a homogeneous polynomial of even degree, then for any $x\in \R^3$, $\phi(-x) = (-1)^n\phi(x) = \phi(x)$. It then follows that 
$$
\int_\Sigma \phi(x_\sigma(t)) d\mu = \int_\Sigma \phi(-x_\sigma(t)) d\mu,
$$
and, hence, $(-x_\Sigma(0), \mu) \sim (x_\Sigma(0), \mu)$.  But then, item~(1) of Theorem~\ref{thm:mthm} says that there is no other pair $(x'_\Sigma(0), \mu')$ that can be output equivalent to $(x_\Sigma(0), \mu)$.  
}	
\end{remark}

{\em Organization of the paper.} In the remainder of the paper, we develop   methods for addressing the ensemble observability problem and prove Theorem~\ref{thm:mthm}.  We will first introduce in Sec.~\ref{sec:prelim} key definitions and notations that will be frequently used throughout the paper. Because our methods rely on the use of representation theory of $\sl(2,\C)$ on the space of homogeneous polynomials (where $\sl(2,\C)$ is the special linear Lie algebra of all $2\times 2$ complex matrices with zero trace), we present in Sec.~\ref{sec:rep} relevant results about such a representation. Then, in Sec.~\ref{sec:proof}, we demonstrate how the representation theory can be used to addressed the ensemble observability problem. The proof of Theorem~\ref{thm:mthm} will also be established along the analysis. We provide conclusions and further discussions in Sec.~\ref{sec:conclusion}. In particular, we will discuss about connections with our earlier work~\cite{chen2019structure} and extensions of the methods developed in the paper to other general ensembles of nonholonomic control systems.

\section{Definitions and Notations}\label{sec:prelim}
We introduce in the section key definitions and notations that will be frequently used throughout the paper.

\subsection{Differential geometry}
For any two smooth vector fields~$f$ and~$g$ on $S^2$, we let $[f, g]$ be the Lie bracket defined as follows:
$$
[f, g](x) := \frac{\partial g}{\partial x} f(x) - \frac{\partial f}{\partial x} g(x)
$$ 
Note that $[f, g]$ is also a vector field on $S^2$. 
Recall that $f_0$ is the drifting vector field and $f_1$, $f_2$ are control vector fields defined in~\eqref{eq:vectorfields}. 
We let~$\g$ be the $\R$-span of $f_0$, $f_1$, and $f_2$. Then, $\g$ is a (real) Lie algebra with the Lie bracket defined above. Note that if $(i,j,k)$ is a cyclic rotation of $(0,1,2)$, then $$[f_i, f_j] = f_k.$$ The above structural coefficients then imply that $\g$ is isomorphic to $\so(3)$ (or simply $\g\approx \so(3)$), where $\so(3)$ is the Lie algebra of~$3\times 3$ real skew-symmetric matrices. We also note that $\so(3)$ is isomorphic $\mathfrak{su}(2)$, i.e., the special unitary Lie algebra (as a real Lie algebra) comprising all $2\times 2$ skew-Hermitian matrices with zero trace. Thus, $\g\approx \mathfrak{su}(2)$ as well.  

For a given vector field $f\in \g$ and a smooth function~$\phi$ on~$S^2$, we let $f\phi$ be another function on $S^2$ defined as follows: 
$$
(f\phi)(x) = \lim_{\epsilon \to 0} \frac{\phi(x + \epsilon f(x)) - \phi(x)}{\epsilon}, \quad \forall x\in S^2.
$$
Note that $(f\phi)(x)$ is nothing but the directional derivative of~$\phi$ along~$f$ at~$x$.  
 
Let $\mathcal{A}$ be the collection of words over the alphabet $\{0,1,2\}$, i.e., $\mathcal{A}$ comprises all finite sequences $\alpha = i_1i_2\cdots i_k$ where each $i_j$ belongs to $\{0,1,2\}$. The length of a word~$\alpha$ is defined to be the total number of indices~$i_j$ in it. 
Next, for a given word $\alpha = i_1\cdots i_k$ and a smooth function~$\phi$ on~$S^2$, we let $$f_\alpha \phi:= f_{i_1}\cdots f_{i_k} \phi.$$ 
Note that if $\alpha = \varnothing$, then we let $f_\alpha \phi := \phi$.

Let $T(\g)$ be the vector space spanned by $f_\alpha$, i.e., each element $\eta$ in $T(\g)$ is a linear combination of finitely many $f_\alpha$ for $\alpha\in \mathcal{A}$. Note that $T(\g)$ can be identified with the space of tensors of~$\g$. Specifically, each $f_\alpha$ can be viewed as a tensor in $\g\otimes \cdots \otimes \g$, where the number of copies of $\g$ matches the length of~$\alpha$.

\subsection{Lie algebra representation}
For an arbitrary real vector space~$V$, we let $V^\C$ be the {\em complexification} of~$V$, i.e., $V^\C$ is a complex vector space comprising all elements $v + \mathrm{i} w$ where $\mathrm{i}$ is the imaginary unit and $v, w$ belong to~$V$. Recall that  $\g \approx \mathfrak{su}(2)$ and, hence, its complexification is given by~\cite[Ch~3.6]{hall2015lie}
$$\g^\C\approx \sl(2,\C),$$ where $\sl(2,\C)$ is the Lie algebra of $2\times 2$ complex matrices with zero trace. 
 
We also recall that $P_n$ is the space of homogeneous polynomials of degree~$n$ in variables $x_1$, $x_2$, and $x_3$. Note that for any $f_i$, with $i = 0,1,2$, and any $p\in P_n$, $f_ip$ belongs to $P_n$. Thus, $P_n$ is closed under directional derivative along any~$f\in \g$.   
We now define a map    
$\pi: \g\times P_n \to P_n$ as follows:
$$
\pi: (f, p) \mapsto \pi(f)p:= fp.
$$ 
The map $\pi$ is in fact a {\em representation} of~$\g$ on $P_n$, i.e., each $\pi(f)$ for, $f\in \g$,    
is an endomorphism of $P_n$ and satisfies the following relationship:  
\begin{equation}\label{eq:representation}
 \pi([f, g]) =  \pi(f)\pi(g) - \pi(g)\pi(f), \quad \forall f, g\in \g.
\end{equation}
We will use $\pi(f) p$ and $fp$ interchangeably. 

Let $P'_n$ be a subspace of $P_n$. We say that $P'_n$ is invariant under $\pi(\g)$ if for any $f\in \g$ and $p\in P'_n$, we have that $\pi(f)p\in P'_n$. Thus, if we let $$\pi': \g\times P'_n \to P'_n$$ be defined by restricting $\pi$ to $\g\times P'_n$, then $\pi'$ is a representation of $\g$ on $P'_n$. We say that $\pi'$ is {\em irreducible} if there does not exist a nonzero, proper subspace $P''_n$ of $P'_n$ such that $P''_n$ is invariant under $\pi(\g)$. 

We further note that the representation $\pi$ can be naturally extended to $\g^\C\times \BP^\C_n$: For any $f, g\in \g$ and any $p, q\in P_n$, let 
\begin{equation}\label{eq:complexificationpi}
\pi(f + \mathrm{i} g) (p + \mathrm{i}q) := \\
(\pi(f)p - \pi(g)q) + \mathrm{i}(\pi(f)q + \pi(g)p).  
\end{equation}
Then, with such an extension, $\pi$ is a representation of $\g^\C$ on $P^\C_n$. 
We will present a few relevant facts about the representation in Sec.~\ref{sec:rep}.

\subsection{Algebra of functions}
Let $\Phi:=\{\phi_i\}^l_{i = 1}$ be a set of functions on $S^2$. 
Let $n_1,\ldots, n_l$ be nonnegative integers. We call $\phi^{n_1}_1\cdots \phi_l^{n_l}$ a {\em monomial}. The {\em degree} of the monomial is $\sum^l_{i = 1}n_i$. Denote by $\cS(\Phi)$ the algebra of generated by~$\Phi$, i.e., each element in $\cS(\Phi)$ is a linear combination of finitely many monomials. Further, we let $\cS_k(\Phi)$ be the subspace of $\cS(\Phi)$ spanned by all monomials of degree~$k$.

\section{Representation on homogeneous polynomials}\label{sec:rep}
 We present in the section a few relevant results (with Prop.~\ref{thm:xixi} the main result) that will be of great use in establishing Theorem~\ref{thm:mthm}. Some of the results are well known. For completeness of presentation, we provide short proofs in the Appendices.   
 
To proceed, we recall that $\mathcal{A}$ is the collection of words over the alphabet $\{0,1,2\}$ and $T(\g)$ is the vector space spanned by all~$f_\alpha$ for $\alpha\in \mathcal{A}$. Now, for a given word $\alpha\in \mathcal{A}$, 
we let $$\kappa(\alpha):= (\kappa_1(\alpha), \kappa_2(\alpha))\in \mathbb{Z}^2$$
where $\kappa_1(\alpha)$ and $\kappa_2(\alpha)$ are defined as follows: 
$$
\left\{
\begin{array}{l}
\kappa_1(\alpha):= \mbox{number of appearances of } ``0" \mbox{ in } \alpha, \\ 
\kappa_2(\alpha):= \mbox{number of appearances of } ``1" \mbox{ and } ``2" \mbox{ in } \alpha
\end{array}
\right.
$$
For example, if $\alpha = 0121$, then $\kappa(\alpha) = (1,3)$.

Next, with a slight abuse of notation, we let $\kappa(f_\alpha) := \kappa(\alpha)$.  
Further, we consider an element $\xi = \sum^n_{i = 1} c_i f_{\alpha_i}$ in $T(\g)$. Suppose that $\kappa(f_{\alpha_i}) = \kappa(f_{\alpha_j})$ for all $i, j\in \{1,\ldots, n\}$; then, we can define without ambiguity that 
$$\kappa(\xi):= \kappa(f_{\alpha_i}), \quad \mbox{for some } i\in \{1,\ldots, n\}.$$ 
Note that if $\kappa(\xi)$ is defined, then the lengths of all the words $\alpha_i$ that are involved in $\xi$ are identical with each other.

We establish in the section the following result:

\begin{proposition}\label{thm:xixi}
There exist nonzero $\xi$ and $\zeta$ in $T(\g)$ such that the following properties are satisfied: 
\begin{enumerate}
\item[(1)] Both $\kappa(\xi)$ and $\kappa(\zeta)$ are well defined. Moreover, 
$$\kappa_1(\xi) > 0, \quad  \mbox{and} \quad \kappa_1(\zeta) = 0.$$ 
\item[(2)] For any $p \in H_n$ with $n \ge 1$,  
$$
\xi p = \zeta p = \lambda p,  
$$ 	
where $\lambda$ is some nonzero constant.
\end{enumerate}	
\end{proposition}

We establish below Prop.~\ref{thm:xixi}. We will explicitly construct $\xi$ and $\zeta$ in Sec.~\ref{ssec:casimir} and show that they satisfy the two items toward the end of the section.  

\subsection{Variations of the Casimir element}\label{ssec:casimir}
Recall that the space $T(\g)$ can be identified with the space of all tensors in $\g^{\otimes k}$ for all $k\ge 0$, i.e., we identify $f_\alpha = f_{i_1}\cdots f_{i_k}$ with $f_{i_1}\otimes \cdots \otimes f_{i_k}$.   
The so-called universal enveloping algebra associated with~$\g$ is defined as follows:  

\begin{definition}
	Let~$J$ be a two sided ideal in $T(\g)$ generated by all $fg - gf - [f, g]$ where $f, g\in \g$. Then, the {\bf universal enveloping algebra} $U(\g)$ is given by the following quotient:   
$$
U(\g):= T(\g)/ J. 
$$    
\end{definition}

We also need the following definition:

\begin{definition}
The {\bf center} $Z(\g)$ of $U(\g)$ is the collection of elements in $U(g)$ that commute with the entire~$U(\g)$, i.e.,  
$$
Z(\g):= \{\eta\in U(\g) \mid \eta' \eta = \eta \eta', \mbox{ for all } \eta' \in U(\g)\}. 
$$
\end{definition}

We present in the following lemma a specific element in $Z(\g)$.  The result is, in fact, well known: 

\begin{lemma}\label{lem:casimir}
Let 
$\eta^*:= \sum^2_{i = 0} f^2_i$. Then, $\eta^*$ belongs to $Z(\g)$.   
\end{lemma}

We provide a proof of the lemma in Appendix-A.

\begin{definition}  
The element~$\eta^* = \sum^2_{i = 0} f^2_i$ is commonly referred to as the {\bf Casimir element}. 
\end{definition}

\begin{remark}{\em 
Note that if an element $\eta$ belongs to $Z(\g)$, then any polynomial in~$\eta$ (i.e., $\sum^n_{k = 0} c_k \eta^k$) belongs to~$Z(\g)$ as well. The converse also holds for the case here. Precisely, it is known~\cite[Ch.~V]{knapp2013lie} that if $\g \approx \so(3) \approx \mathfrak{su}(2)$, then the center $Z(\g)$ is exactly the space of all polynomials in~$\eta^*$. We further note that for a general (complex) semi-simple Lie algebra, the center of the associated universal enveloping algebra can be characterized via the Harish-Chandra isomorphism~\cite[Theorem~5.44]{knapp2013lie}. 
}
\end{remark}

However, note that if we treat the Casimir element $\eta^*$ as an element in $T(\g)$, then $\kappa(\eta^*)$ is not well defined. To see this, we simply note that $$\kappa(f^2_0) = (2, 0) \quad  \mbox{and} \quad \kappa(f^2_1) = \kappa(f^2_2) = (0,2).$$ 
We thus aim to find elements~$\xi$ and~$\zeta$ in $T(\g)$ that satisfy the following two conditions: 
\begin{enumerate}
\item[\em (1)] Both $\kappa(\xi)$ and $\kappa(\zeta)$ are well defined and satisfy item~(1) of Prop.~\ref{thm:xixi}.    
\item[\em (2)] The two elements $\xi$ and $\zeta$ are the same as the Casimir element $\eta^*$ when they are treated as elements in $U(\g)$, i.e., all the three elements are equivalent modulo the ideal~$J$ (we will simply write $\xi\equiv \zeta \equiv \eta^*$).  
\end{enumerate}     
One way to find such elements~$\xi$ and~$\zeta$ is to use the commutator relations: $[f_i, f_j] = f_k$ where $(i,j,k)$ is a cyclic rotation of $(0,1,2)$.  
We have the following result:   
 
\begin{lemma}\label{lem:defxizeta}   
Let $\xi, \zeta \in T(\g)$ be defined as follows:
$$
\left\{
\begin{array}{l}
\xi:= f_0f_1f_2 + f_1f_2f_0 + f_2f_0f_1 \\
\qquad \qquad\qquad - f_0f_2f_1  - f_1f_0f_2  - f_2f_1f_0, \vspace{3pt} \\
\zeta:= 3(f_1f_2f_1f_2 + f_2f_1f_2f_1) \\ 
\qquad \qquad\qquad -2(f_1f^2_2f_1 + f_2f_1^2 f_2) - (f^2_1f^2_2 + f^2_2f^2_1).  
\end{array}
\right.
$$
Then,  $\xi\equiv \zeta\equiv\eta^*$ with $\kappa(\xi) = (1, 2)$ and $\kappa(\zeta) = (0,4)$.
\end{lemma}

\begin{proof}{Proof.}
The lemma follows directly from computation. Specifically, we note that 
$$
\left\{
\begin{array}{l}
	f^2_0 \equiv f_0f_1f_2 - f_0f_2f_1, \vspace{3pt} \\
	f^2_1 \equiv f_1f_2f_0 - f_1f_0f_2, \vspace{3pt} \\
	f^2_2 \equiv f_2f_0f_1 - f_2f_1f_0.
\end{array}
\right.
$$
The element $\xi$ is then obtained by replacing $f^2_i$ for $i = 0,1,2$ in~$\eta^*$ with the terms on the right hand side of the above expression. Further, by replacing each~$f_0$ in the expression of~$\xi$ with $(f_1f_2 - f_2f_1)$, we obtain~$\zeta$. 
\end{proof}

%Note that there are many other (in fact, infinitely many) different $\xi$ and $\zeta$ in $T(\g)$ that satisfy the above two items.  
%These elements can again be obtained by recursively applying the commutator relations $[f_i,f_j] = f_k$ for $(i,j,k)$ a cyclic rotation of $(0,1,2)$.  

Toward the end of the section, we will show that the two elements~$\xi$ and~$\zeta$ defined in Lemma~\ref{lem:defxizeta} satisfy item~(2) of Prop.~\ref{thm:xixi}. For that, we need to have a few preliminaries about irreducible representation of~$\sl(2,\C)$, some of which will further be used in the proof of Theorem~\ref{thm:mthm}. This will be done in the next subsection.

\subsection{Irreducible representation of $\sl(2, \C)$}
Recall that $\BP_n$ is the (real) vector space of all homogeneous polynomials of degree~$n$ in variables $x_1$, $x_2$, and $x_3$. The space $\BP_n$ is closed under directional derivative along any vector field $f\in \g$. The map $\pi: \g\times \BP_n \to \BP_n$ defined by
$$
\pi: (f, p) \mapsto \pi(f)p := fp
$$ 
is a Lie algebra representation of~$\g$ on $\BP_n$. 
We also recall that $\BP^\C_n$ is the complexification of $\BP_n$, i.e., $\BP^\C_n$ is the space of homogeneous polynomials in (real) variables $x_1$, $x_2$, $x_3$ with complex coefficients.

One can extend $\pi$ to $\g^\C \times P^\C_n$ using~\eqref{eq:complexificationpi} so that $\pi$ is now a representation of $\g^\C$ on $\BP^\C_n$.  
Note that $\g^\C \approx \sl(2, \C)$. 
Representation of $\sl(2, \C)$ is extensively investigated in the literature~\cite{hall2015lie,knapp2013lie,humphreys2012introduction}. We review in the subsection only a few basic facts that are relevant to the paper.

To this end, we define a triplet $(h, e_+, e_-)$ of elements in $\g^\C$ using the three elements $\{f_i\}^2_{i =0}$ from~$\g$ as follows:
\begin{equation}\label{eq:triplet}
h:= 2\mathrm{i} f_0, \quad e_+ := f_1 + \mathrm{i} f_2, \quad e_-:= -f_1 + \mathrm{i} f_2,
\end{equation}
Then, by computation, we have the following standard commutator relationship for the triplet $(h, e_+, e_-)$ in~$\sl(2,\C)$:  
$$
[h, e_+] = 2 e_+, \quad [h, e_-] = -2e_-, \quad [e_+, e_-] = h. 
$$ 
Denote by $\C h$, $\C e_+$, and $\C e_-$ the vector spaces (over $\C$) spanned by $h$, $e_+$, and $e_-$, respectively. Then, $\C h$ is known as a {\em Cartan subalgebra} of $\sl(2, \C)$ while $\C e_+$ and $\C e_-$ are the two {\em root spaces}. Recall that an arbitrary representation $\pi: \sl(2, \C)\times V \to V$ is irreducible if there does not exist a nonzero, proper subspace~$V'$ of~$V$ such that $\pi(V') \subseteq V'$. The following result is well-known (see, for example,~\cite{hall2015lie}) for finite-dimensional irreducible representations of~$\sl(2, \C)$:

\begin{lemma}\label{lem:irredrep}
Let $\pi: \sl(2, \C) \times V \to V$ be an arbitrary irreducible representation of $\sl(2, \C)$ on a (complex) vector space~$V$ of dimension~$(n + 1)$ for $n\ge 0$. Then, $V$ can be decomposed as a direct sum of one-dimensional subspaces $V = \oplus^{n}_{k = 0}V_{n - 2k}$,  which satisfy the following conditions:  
$$
\pi(e_+) V_{n - 2k} = V_{n - 2k + 2}, \quad \pi(e_-) V_{n - 2k} = V_{n - 2k - 2}.
$$
Moreover, for any $v\in V_{n - 2k}$,  $\pi(h)v = (n - 2k)v$. 
\end{lemma}

\begin{definition}\label{def:weight}
The subspaces $V_{n - 2k}$ in the above lemma are {\bf weight spaces},  and the integers $(n - 2k)$ are {\bf weights}. The weight~$n$ (i.e., $k = 0$) is called the {\bf highest weight} and,  correspondingly, any nonzero vector~$v$ in $V_n$ is called a {\bf highest weight vector}. 
\end{definition}

Note that by Lemma~\ref{lem:irredrep}, if~$v$ is a highest weight vector (of weight~$n$), then the set of vectors $\{v, \pi(e_-)v, \cdots, \pi^n(e_-)v\}$ is a basis of~$V$. Each one-dimensional weight space $V_{n - 2k}$ is spanned by the vector $\pi^k(e_-)v$. Conversely, we have the following fact:

\begin{lemma}\label{lem:highestweight}
	Let $\pi: \sl(2,\C)\times V\to V$ be an arbitrary representation (not necessarily irreducible). Suppose that there is a nonzero vector~$v\in V$ and an integer $n \ge 0$ such that 
	$$
	\pi(h) v = n v \quad \mbox{and} \quad \pi(e_+)v = 0;
	$$
	then, the subspace $V'$ spanned by $\{v, \pi(e_-)v, \cdots, \pi^n(e_-)v\}$ is an invariant subspace of $V$ under $\pi(\sl(2, \C))$. 
	Let $\pi'$ be defined by restricting $\pi$ to $\sl(2, \C)\times V'$, then $\pi'$ is an irreducible representation of $\sl(2,\C)$ on~$V'$ with~$n$ the highest weight and~$v$ a highest weight vector. 
\end{lemma}

The above lemma is an application of the Theorem of Highest Weight~\cite[Theorem~5.5]{knapp2013lie}. 

We now return to the representation  $\pi: \g^\C\times \BP^\C_n \to \BP^\C_n$. We will see soon that $\pi$ is not irreducible. But, by the unitarian trick (see, for example,~\cite[Theorem~5.29]{knapp2013lie}), any finite-dimensional representation of a complex semi-simple Lie algebra is completely reducible. Specifically, we first recall that $\BH_n$ is the (real) vector space of harmonic homogeneous polynomials of degree~$n$. We let $\BH^\C_n$ be its complexification. 
Then, we have the following decomposition:  
$$\BP^\C_n =  \BH^\C_{n} \oplus \|x\|^2\BH^\C_{n - 2} \oplus \cdots \oplus \|x\|^{2\lfloor \nicefrac{n}{2} \rfloor} \BH^\C_{n -2 \lfloor \nicefrac{n}{2} \rfloor},$$ 
where $\|x\|^2= \sum^3_{i =1} x^2_i$. The following fact is well-known~\cite[Ch.~17]{hall2013quantum}:
 
\begin{lemma}\label{prop:harmonicpolynomial}
	The subspace $\|x\|^{2k}\BH^\C_{n - 2k}$ is invariant under $\pi(\g^\C)$ for any $k = 0,\ldots, \lfloor \nicefrac{n}{2} \rfloor$.  
	Define $$\pi_k: \g^\C \times \|x\|^{2k}\BH^\C_{n-2k}\to \|x\|^{2k}\BH^\C_{n-2k}$$ by restricting~$\pi$ to $\g^\C \times \|x\|^{2k}\BH^\C_{n-2k}$. Then, $\pi_k$ is an irreducible representation with $2(n - 2k)$ the highest weight and 
	$$
	p^*_{k} := \|x\|^{2k} (x_1 + \mathrm{i} x_2)^{n - 2k}  
	$$
	a highest weight vector.  
\end{lemma}

We provide a proof of the lemma in Appendix-B.

\subsection{Proof of Proposition~\ref{thm:xixi}} \label{sec:proofofxixi}
We establish in the subsection Prop.~\ref{thm:xixi}. 
With slight abuse of notation, we will now let $$\pi: \g\times H_n\to H_n$$ be the representation of~$\g$ on~$H_n$. 
By Lemma~\ref{prop:harmonicpolynomial}, $\pi$ is irreducible.  
The map $\pi$ can naturally be extended to $T(\g) \times H_n$, which we have implicitly used in the section. Specifically, for any $p\in H_n$ and any $\eta\in T(\g)$, we define $\pi(\eta)p := \eta p$. 
Further, note that the relationship~\eqref{eq:representation} which we reproduce below:
$$
\pi([f, g]) =  \pi(f)\pi(g) - \pi(g)\pi(f), \quad \forall f, g\in \g.
$$ 
allows us to pass the map~$\pi$ to the quotient $U(\g)\times H_n$, i.e., if two elements~$\eta$ and~$\eta'$ in $T(\g)$ are equivalent (i.e., $\eta\equiv \eta'$), then $\eta p = \eta' p$ for any $p\in H_n$.   

Recall that $\eta^* = \sum^2_{i = 0} f^2_i $ is the Casimir element. Let $\xi$ and $\zeta$ be defined in Lemma~\ref{lem:defxizeta} and we have that $\eta^*\equiv \xi\equiv \zeta$. Then, by the above arguments, 
\begin{equation}\label{eq:xizetaeigen}
\eta^* p = \xi p = \zeta p, \quad \forall p\in H_n. 
\end{equation}
The following fact is a consequence of Shur's Lemma (see, for example, Lemma~1.69 in~\cite{knapp2013lie}):

\begin{lemma}\label{lem:diagonaloperator}
The Casimir element $\eta^*$ acts on $H_n$ as a scalar multiple of the identity operator. Specifically, for any $p\in H_n$, we have that
$$
\eta^* p = -n(n + 1) p.
$$   	
\end{lemma}

We provide a proof of the lemma in Appendix-C. 

Prop.~\ref{thm:xixi} then follows from Lemmas~\ref{lem:defxizeta} and~\ref{lem:diagonaloperator}. \hfill{\qed}

%%%%%%%%%%%%%%%%%%%%%%%%%%%%%% Section 4 %%%%%%%%%%%%%%%%%%%%%%%%%%%%%%

\section{Analysis and Proof of Theorem~\ref{thm:mthm}}\label{sec:proof}
We establish in the section Theorem~\ref{thm:mthm}. The proof will be built upon two relevant facts, Prop.~\ref{prop:ppp} and Prop.~\ref{prop:qqq}, which will be presented and established in Subsections~\ref{ssec:piecewiseconstant} and~\ref{ssec:generationconstant}, respectively. We will prove Theorem~\ref{thm:mthm} in Subsection~\ref{ssec:proofmain}. 

\subsection{Analysis of output equivalent pairs}\label{ssec:piecewiseconstant}
Recall that two pairs $(x_\Sigma(0), \mu)$ and $(x'_\Sigma(0),\mu')$ are output equivalent if for any integrable control input $u(t)$, the two outputs $y(t)$ and $y'(t)$ with respect to the two pairs are identical with each other (See Def.~\ref{def:outputequiv}). We establish below the following result: 

\begin{proposition}\label{prop:ppp}
	Let the observation function $\phi$ of system~\eqref{eq:system} be nonzero and belong to $H_n$ for $n \ge 1$. If $(x_\Sigma(0),\mu) \sim (x'_\Sigma(0),\mu')$, then for any $p\in H_n$,
	$$
	p(x_\sigma(0))\rho(\sigma) = p(x'_\sigma(0))\rho'(\sigma), \quad \forall \sigma\in \Sigma,
	$$ 
	where $\rho$ and $\rho'$ are the density functions associated with $\mu$ and $\mu'$, respectively.  
\end{proposition}

To establish the proposition, we need to have a few preliminary results. To proceed, we recall that for an element $f_\alpha\in T(\g)$, we have that $$\kappa(f_\alpha) = (\kappa_1(f_\alpha), \kappa_2(f_\alpha))$$ where $\kappa_1(f_\alpha)$ (resp. $\kappa_2(f_\alpha)$) counts the number of ``$0$'' (resp. ``$1$'' and ``$2$'') in the word $\alpha$ over the alphabet $\mathcal{A} = \{0,1,2\}$. For convenience, we introduce the following notation: 
$$
\sigma^{\kappa(f_\alpha)} := \sigma_1^{\kappa_1(f_\alpha)} \sigma_2^{\kappa_2(f_\alpha)}, \quad \forall  \sigma \in \Sigma,   
$$  
which is a monomial in variables $\sigma_1$ and $\sigma_2$. 
We first have the following fact: 

\begin{lemma}\label{lem:outputequivalent}
Let $\phi$ be any smooth observation function.  
If $(x_\Sigma(0), \mu)\sim (x'_\Sigma(0), \mu')$, then 	for any~$f_\alpha$ with $\alpha\in \mathcal{A}$, 
\begin{equation}\label{eq:skimmy}
	\int_\Sigma \sigma^{\kappa(f_\alpha)}  (f_\alpha \phi)(x_\sigma(0)) d\mu = \int_\Sigma  \sigma^{\kappa(f_\alpha)} (f_\alpha\phi) (x'_\sigma(0)) d\mu'.  
\end{equation}
\end{lemma}

\begin{proof}{Proof.}
Let $n$ be an arbitrary nonnegative integer number. We prove the lemma for any word $\alpha$ of length~$n$. The arguments used in the proof will be similar to the one used in~\cite{chen2019structure}: We will appeal to the class of piecewise constant control inputs to establish~\eqref{eq:skimmy}. 
    
	Define a piecewise constant control input $u(t)$ as follows: First, let $0 < t_1 < \cdots < t_n$ be switching times. Then, we let $u(t):= (u_{i_1}, u_{i_2})$ for $t\in [t_{i-1}, t_i)$ where $t_0 := 0$. Next, for ease of notation, we define for each $i = 1,\ldots, n$, the duration $\tau_i:= t_{i} - t_{i - 1}$ and the corresponding vector field over the period $[t_{i - 1}, t_i)$:	
	\begin{equation}\label{eq:deftildef}
	\tilde f_i :=  \sigma_1 f_0 + \sigma_2 ( u_{i_1}f_1 + u_{i_2} f_2).   
	\end{equation}
	
	We further introduce the following notation: For an arbitrary differential equation $\dot x(t) = f(x(t))$, we let $e^{tf}x(0)$ be the solution of the equation at time~$t$ with $x(0)$ the initial state. In the context here, we have that for any individual system-$\sigma$ with $\sigma \in \Sigma$, the following hold with respect to the piecewise constant control input:
	$$
	\left\{
	\begin{array}{l}
	x_\sigma(t_n) = e^{\tau_n \tilde f_n}\cdots e^{\tau_1 \tilde f_1} x_\sigma(0), \vspace{3pt}\\ 
	x'_\sigma(t_n) = e^{\tau_n \tilde f_n}\cdots e^{\tau_1 \tilde f_1} x'_\sigma(0).
	\end{array} 
	\right. 
	$$
	Thus, if $(x_\Sigma(0), \mu) \sim (x'_\Sigma(0), \mu')$, then for any $\tau_i$ with $i = 1,\ldots, n$, the following holds: 
	\begin{equation*}
	\int_\Sigma \phi\left (e^{\tau_n \tilde f_n}\cdots e^{\tau_1 \tilde f_1} x_\sigma(0)\right )d\mu  
	= \int_\Sigma \phi\left (e^{\tau_n \tilde f_n}\cdots e^{\tau_1 \tilde f_1} x'_\sigma(0)\right )d\mu'.
	\end{equation*}   
	
	We next take partial derivative $\nicefrac{\partial^n}{\partial \tau_1\cdots \partial \tau_n}$ on both sides of the above expression and let them be evaluated at $\tau_1 = \cdots = \tau_n = 0$. Then, by computation, we obtain
	$$
	\int_\Sigma (\tilde f_1 \cdots \tilde f_n \phi)(x_\sigma(0)) d\mu = \int_\Sigma (\tilde f_1 \cdots \tilde f_n \phi)(x'_\sigma(0)) d\mu'. 
	$$
	Note that by~\eqref{eq:deftildef}, each $\tilde f_i$ depends on $(u_{i_1}, u_{i_2})$ and the above expression holds for all $(u_{i_1}, u_{i_2})\in \R^2$ and for all $i= 1,\ldots, n$. Also, note that by expanding each $\tilde f_i$ using~\eqref{eq:deftildef}, we have that $\tilde f_1 \cdots \tilde f_n \phi$ is a linear combination of $\sigma^{\kappa(f_\alpha)} f_\alpha \phi$ for~$\alpha$ any word of length~$n$. 
	It then follows that~\eqref{eq:skimmy} holds. 
\end{proof}

A set of functions $\{\psi_i\}^n_{i = 1}$ on $\Sigma$ is said to {\em separate points} if for any two distinct points $\sigma$ and $\sigma'$ in $\Sigma$, there exists a function $\psi_i$ out of the set such that $\psi_i(\sigma) \neq \psi_i(\sigma')$.  
We also recall that by Lemma~\ref{lem:defxizeta}, $$\kappa(\xi) = (1,2) \quad \mbox{and} \quad  \kappa(\zeta) = (0,4).$$  
We define monomials $m_\xi$ and $m_\zeta$ in variables $\sigma_1$ and $\sigma_2$ as follows:
\begin{equation}\label{eq:pxipzeta}
\left\{
\begin{array}{l}
m_\xi(\sigma):= \sigma^{\kappa(\xi)} = \sigma_1\sigma^2_2, \vspace{3pt} \\
m_\zeta(\sigma) := \sigma^{\kappa(\zeta)} = \sigma^4_2.
\end{array}
\right.
\end{equation} 
We next have the following fact:

\begin{lemma}\label{lem:sssssss}
The set $\{m_\xi, m_\zeta\}$ separates points and, moreover, $m_\zeta$ is everywhere nonzero.  
\end{lemma}

\begin{proof}{Proof.}
First, we recall that $\Sigma = [a_1, b_1]\times [a_2, b_2]$ with $0 < a_2 < b_2$. Thus, for any $\sigma = (\sigma_1,\sigma_2)\in \Sigma$, we have that $\sigma_2\in [a_2, b_2]$ and, hence, $m_\zeta$ is everywhere nonzero. 
	Next, we let $\sigma = (\sigma_1, \sigma_2)$ and $\sigma' = (\sigma'_1, \sigma'_2)$ be two distinct points in $\Sigma$. If $\sigma_2 \neq \sigma'_2$, then $m_\zeta(\sigma) \neq m_\zeta(\sigma')$. If $\sigma_2 = \sigma'_2$, then $\sigma_1 \neq \sigma'_1$ and, hence, $m_\xi(\sigma) \neq m_\xi(\sigma')$. 
\end{proof}

With the above lemmas at hand, we prove Prop.~\ref{prop:ppp}:

\begin{proof}{Proof of Prop.~\ref{prop:ppp}.}  
  Recall that $U(\g)$ is the universal enveloping algebra associated with~$\g$. 
Let $p$ be any nonzero polynomial in $H_n$ and 
$$H'_n:= U(\g) p = \{\eta p \mid \eta\in U(\g) \}.$$  
Let $\pi: \g\times H_n \to H_n$  be the representation defined in Sec.~\ref{sec:proofofxixi}, i.e.,  $$\pi: (f, p)\in \g\times H_n \mapsto \pi(f)p:= fp.$$     
Because $H_n$ is closed under $\pi(\g)$, $H'_n$ is a subspace of $H_n$. Also, note that by the definition, $H'_n$ itself is closed under $\pi(\g)$. Thus, by the fact that $\pi$ is an irreducible representation (Lemma~\ref{prop:harmonicpolynomial}), we must have that $H_n = H'_n = U(\g) p$. Since $U(\g)$ is spanned by $f_{\alpha}$ for $\alpha\in \mathcal{A}$ and $\dim H_n= 2n + 1$, there exist $f_{\alpha_i}$, for $i = 1,\ldots, 2n + 1$, such that $f_{\alpha_i} p$ form a basis of $H_n$. For convenience, we let 
$$
p_i := f_{\alpha_i} p, \quad \forall i =1,\ldots, 2n + 1.
$$ 

Let $(x'_\Sigma(0), \mu') \sim (x_\Sigma(0), \mu)$ be two output equivalent profiles. Let $\rho$ and $\rho'$ be the density functions associated with $\mu$ and $\mu'$, respectively. 
By Lemma~\ref{lem:outputequivalent}, we have that for any $i = 1,\ldots, 2n + 1$ and any word $\alpha$ over the alphabet $\{0,1,2\}$, the following holds:  
$$
\int_\Sigma \sigma^{\kappa(f_\alpha) + \kappa(f_{\alpha_i})} (f_\alpha p_i)(x_\sigma (0)) d\mu  = \int_\Sigma \sigma^{\kappa(f_\alpha) + \kappa(f_{\alpha_i})} (f_\alpha p_i)(x'_\sigma (0)) d\mu'. 
$$
The above equality can be further strengthened by replacing $f_\alpha$ with any $\eta \in T(\g)$ such that $\kappa(\eta)$ is well defined, i.e.,   
\begin{equation}\label{eq:kuanrong}
\int_\Sigma \sigma^{\kappa(\eta) + \kappa(f_{\alpha_i})} (\eta p_i)(x_\sigma (0)) d\mu  = \int_\Sigma \sigma^{\kappa(\eta) + \kappa(f_{\alpha_i})} (\eta p_i)(x'_\sigma (0)) d\mu'.
\end{equation}

Now, let $\xi$ and $\zeta$ be defined in Lemma~\ref{lem:defxizeta}. Then, by~\eqref{eq:xizetaeigen} and Lemma~\ref{lem:diagonaloperator},  we have that for any~$N \ge 0$ and $i = 1,\ldots, 2n + 1$,  
\begin{equation}\label{eq:NNNNN}
\xi^N p_i = \zeta^N p_i = \lambda^N p_i,   
\end{equation}
with $\lambda := -n(n + 1)$. 
Thus, by replacing $\eta$ in~\eqref{eq:kuanrong} with $\xi^N$ or $\zeta^N$ and by omitting $\lambda^N$ on both sides, we obtain the following equalities:  
\begin{equation}\label{eq:innerproduct}
\left\{
\begin{array}{l}
\displaystyle \int_\Sigma m^N_\xi(\sigma) \psi_i(\sigma) d\sigma = \displaystyle \int_\Sigma m^N_\xi(\sigma)  \psi'_i(\sigma) d\sigma, \vspace{3pt}\\
\displaystyle \int_\Sigma m^N_\zeta(\sigma) \psi_i(\sigma) d\sigma = \displaystyle \int_\Sigma m^N_\zeta(\sigma) \psi'_i(\sigma) d\sigma,
\end{array}  
\right. 
\end{equation}
where $m_\xi$, $m_\zeta$ are monomials given by~\eqref{eq:pxipzeta} and $\psi_i$, $\psi'_i$ are defined as follows: 
$$
\left\{
\begin{array}{l}
\psi_i(\sigma):= \sigma^{\kappa(f_{\alpha_i})}p_i(x_\sigma(0))\rho(\sigma), \vspace{3pt} \\
\psi'_i(\sigma):= \sigma^{\kappa(f_{\alpha_i})}p_i(x'_\sigma(0))\rho'(\sigma),
\end{array}
\right. 
$$
for all $i = 1,\ldots, 2n + 1$.

Let ${\rm C}^0(\Sigma)$ be the space of continuous functions on $\Sigma$ and 
${\rm L}^2(\Sigma)$ be the space of square integrable functions $\psi$ on~$\Sigma$, i.e., 
$\int_\Sigma \|\psi\|^2d\sigma < \infty$. Note that ${\rm L}^2(\Sigma)$ is an inner-product space: For any $\psi$ and $\psi'$ in ${\rm L}^2(\Sigma)$, we let their inner-product be defined as follows: 
$$
\langle \psi, \psi'\rangle_{{\rm L}^2} := \int_\Sigma \psi(\sigma)\psi'(\sigma) d\sigma. 
$$ 
By Lemma~\ref{lem:sssssss}, the set $\{m_\xi, m_\zeta\}$ separates points and, moreover, $m_\zeta$ is everywhere nonzero on $\Sigma$. Thus, by the Stone-Weierstrass Theorem (see, for example,~\cite{rudin1976principles}), the algebra generated by $m_\xi$ and $m_\zeta$ is dense in ${\rm C}^0(\Sigma)$. Furthermore, since $\Sigma$ is compact, ${\rm C}^0(\Sigma)$ is dense in ${\rm L}^2(\Sigma)$. 
It then follows from~\eqref{eq:innerproduct} that $\psi_i(\sigma) = \psi'_i(\sigma)$ for almost all $\sigma\in \Sigma$. Since $\psi_i$ and $\psi'_i$ are continuous on~$\Sigma$, the two functions are identical:  
$$
\sigma^{\kappa(f_{\alpha_i})}p_i(x_\sigma(0))\rho(\sigma) = \sigma^{\kappa(f_{\alpha_i})}p_i(x'_\sigma(0))\rho'(\sigma), \quad \forall \sigma\in \Sigma.
$$ 
Furthermore, by continuity of $p_i$ and $\rho$, we obtain that
$$
p_i(x_\sigma(0))\rho(\sigma) = p_i(x'_\sigma(0))\rho'(\sigma), \quad \forall \sigma\in \Sigma.
$$ 
Note that the above holds for all $i = 1,\ldots, 2n+1$. Since $\{p_i\}^{2n + 1}_{i = 1}$ is a basis of $H_n$, we conclude that Prop.~\ref{prop:ppp} holds. 
\end{proof}

\begin{remark}{\em 
	Note that the two items of Prop.~\ref{thm:xixi} are instrumental in establishing Prop.~\ref{prop:ppp}: Item~(1) of Prop.~\ref{thm:xixi} guarantees that Lemma~\ref{lem:sssssss} is satisfied while item~(2) of Prop.~\ref{thm:xixi} guarantees that~\eqref{eq:innerproduct} holds.  
	}
\end{remark}

\subsection{Constant function in quadratic form}\label{ssec:generationconstant}

If there were a harmonic homogeneous polynomial $p^*$ of positive degree such that $p^*$ is a nonzero constant function over the entire $S^2$, then by Prop.~\ref{prop:ppp}, we obtain that $p^*\rho = p^*\rho'$ and, hence, $\rho = \rho'$ (i.e., $\mu = \mu'$).

However, such harmonic homogeneous polynomial $p^*$ does not exist. Nevertheless, we show in the section that there is a quadratic form in $p\in H_n$ (for any $n\ge 1$) which is exactly a nonzero constant function on $S^2$. We make the statement precise below.  

We first recall that for a given set of functions $\Phi:= \{\phi_i\}^l_{i = 1}$ on~$S^2$, we use $\mathcal{S}(\Phi)$ to denote the algebra generated by the set~$\Phi$, i.e., it comprises all linear combinations of finitely many monomials $\phi^{n_1}_1\cdots \phi^{n_l}_l$. Also, recall that $\cS_2(\Phi)$ is the space of quadratic forms in $\phi_i$ for $i = 1,\ldots, l$, i.e., $\cS_2(\Phi)$ is spanned by $\phi_i\phi_j$ for $1\le i \le j \le l$.

We now let $\Phi = \{p_i\}^{2n + 1}_{i = 1}$ be an arbitrary basis of $H_n$.  
Note that each $q\in \cS_2(\Phi)$ is a quadratic form in $p_i$ and each $p_i$ is a homogeneous polynomial of degree~$n$ in $x_1$, $x_2$, and $x_3$. Thus, each $q\in \cS_2(\Phi)$ is a homogeneous polynomial of degree~$2n$ in $x_1$, $x_2$, and $x_3$.

We further let~${\bf 1}_{S^2}$ be the constant function that takes value~$1$ everywhere on $S^2$, i.e., $${\bf 1}_{S^2}(x) := 1, \quad \forall x\in S^2.$$
We establish below the following result:  

\begin{proposition}\label{prop:qqq}
For any basis $\Phi$ of $H_n$, $\cS_{2}(\Phi)$ contains $\|x\|^{2n}$ and, hence, the constant function ${\bf 1}_{S^2}$. 	
\end{proposition}

\begin{remark}\label{rmk:sames2}
{\em We note that for any two bases $\Phi = \{\phi_k\}^{2n + 1}_{k = 1}$ and $\Phi' = \{\phi'_k\}^{2n + 1}_{k = 1}$ of $H_n$,  $$\cS_2(\Phi) = \cS_2(\Phi').$$ 
This holds because $\mathcal{S}_2(\Phi)$ and $\mathcal{S}_2(\Phi')$ are spanned by $\phi_i\phi_j$ and $\phi'_i\phi'_j$, respectively, each $\phi_i\phi_j$ (resp. $\phi'_i\phi'_j$) can be expressed as a linear combination of $\phi'_i\phi'_j$ (resp. $\phi_i\phi_j$). 
More specifically, since $\Phi'$ is a basis of $H_n$ and $\phi_i, \phi_j\in H_n$, there are real coefficients $c_{i,k}$ and $c_{j,k}$, for $k = 1,\ldots, 2n + 1$ such that $\phi_i = \sum_{k = 1}^{2n + 1} c_{i,k} \phi'_k$ and $\phi_j = \sum_{k = 1}^{2n + 1} c_{j,k}\phi'_i$. It then follows that
$$
\phi_i \phi_j = \sum_{1\le k, k'\le 2n + 1}c_{i,k}c_{j,k'} \phi'_i\phi'_j.
$$
By the same argument, we can express $\phi'_i\phi'_j$ as a certain linear combination of $\phi_i\phi_j$ as well. Thus, by the above arguments, we only need to prove Prop.~\ref{prop:qqq} for a particular basis~$\Phi$ of $H_n$. We will make a choice of $\Phi$ later in~\eqref{eq:choicefphi}.}
\end{remark}

 Before proving Prop.~\ref{prop:qqq}, we take an example for illustration of the statement:

\begin{example}\label{exmp:identity}
We demonstrate Prop.~\ref{prop:qqq} for $n = 1, 2, 3$: 
\begin{enumerate}
\item[(1)] If $n = 1$, then $H_1$ is spanned by $\{x_1,x_2,x_3\}$, so $\mathcal{S}_2(\Phi)$ contains $\|x\|^2$. 
\item[(2)] If $n = 2$, then a basis of $H_2$ is given by
$$
p_1:= x^2_1 - x^2_2, \,\,\,\,\  p_2:= x^2_2 - x^2_3,  \,\,\,\,\
p_3:= x_1x_2,\,\,\,\,\ p_4:= x_1x_3,\,\,\,\,\ p_5:= x_2x_3.
$$ 
By computation, we obtain that
$$
\|x\|^4 = p_1^2 + p_2^2 + p_1p_2 
+ 2 \left (p^2_3 + p^2_4 + p^2_5 \right ). 
$$
\item[(3)] If $n = 3$, then a basis of $H_3$ is given by
$$
\begin{array}{l}
p_1:= x_1(2x^2_1 - 3x^2_2 - 3x^2_3), \vspace{3pt} \\
p_2:= x_2(2x^2_2 - 3x^2_1 - 3x^2_3), \vspace{3pt} \\ 
 p_3:=x_3(2x^2_3 - 3x^2_1 - 3x^2_2), \vspace{3pt} \\
 p_4:=x_1(x^2_2 - x^2_3), \,\,\,\,\ p_5:=x_2(x^2_1 - x^2_3), \vspace{3pt} \\
 p_6:= x_3(x^2_1 - x^2_2), \,\,\,\,\ p_7:= x_1x_2x_3.
 \end{array}
$$
By computation, we obtain that
$$
\|x\|^6 = \frac{1}{4}\left ( p^2_1 + p^2_2 + p^2_3 \right ) + \frac{15}{4}\left (p^2_4 + p^2_5 + p^2_6 \right ) + 15p^2_7.
$$
\end{enumerate}
\end{example}

We establish below Prop.~\ref{prop:qqq}. 
There are several different approaches for proving the result. The approach we present below utilizes again the representation theory of $\sl(2,\C)$. One can also use the Addition Theorem for spherical harmonics~\cite[Ch.~12]{arfken2005mathematical} to prove the result. For that, we refer the reader to Appendix-F for detail.   

To proceed, we first recall that by Lemma~\ref{prop:harmonicpolynomial}, the polynomial $(x_1+ \mathrm{i} x_2)^{n}$ is a highest weight vector (with the highest weight being~$2n$) associated with the irreducible representation $\pi: \g^\C \times H^\C_n \to H^\C_n$. Let $h$, $e_+$, and $e_-$ be defined in~\eqref{eq:triplet}. We next define 
\begin{equation}\label{eq:defphikkk}
p_{k}(x) := \pi^k(e_-) (x_1+ \mathrm{i} x_2)^{n}, \quad \forall k = 0,\ldots, 2n. 
\end{equation}
Then, by Lemma~\ref{lem:irredrep}, each $p_k$ is a weight vector and 
\begin{equation}\label{eq:relationship1}
\pi(h) p_k = (2n - 2k) p_k.
\end{equation}  
 It should be clear from the definition that $\pi(e_-)p_k = p_{k + 1}$ for all $k = 0,\ldots, 2n-1$. Conversely, for any $k = 1,\ldots, 2n$, the following holds (see, for example,~\cite[Ch.~17]{hall2013quantum}): 
\begin{equation}\label{eq:reverse}
\pi(e_+) p_{k} = k(2n - k + 1) p_{k-1}.  
\end{equation}
Furthermore, we have the following fact:

\begin{lemma}\label{lem:niu}
For any $k=0,\ldots, n$,  
\begin{equation}\label{eq:hshshs}
p_{2n - k} = (-1)^{n - k} \frac{(2n - k)!}{k!} \bar p_k, 
\end{equation}	
where $\bar p_k$ is the complex conjugate of $p_k$. 
\end{lemma}

We provide a proof in Appendix-D. 
Note, in particular, that by~\eqref{eq:hshshs}, $p_n = \bar p_n$ and, hence, $p_n$ is real.

With the $p_k$ define in~\eqref{eq:defphikkk}, we now let  
\begin{equation}\label{eq:choicefphi}
\Phi:= \{p_k\}^{2n}_{k = 0}.
\end{equation} 
By Lemma~\ref{lem:highestweight}, $\Phi$ is a basis of $H^\C_n$ over $\C$. 
Let $\mathcal{S}^\C_2(\Phi)$ be the complexification of $\cS_2(\Phi)$, i.e., $\mathcal{S}^\C_2(\Phi)$ is the space of all quadratic forms in $p_k$ with complex coefficients. 
To establish Prop.~\ref{prop:qqq}, it now suffices to show that the monomial $\|x\|^{2n}$ is contained in 
$\mathcal{S}^\C_2(\Phi)$ (note that if this is the case, then $\|x\|^{2n}$ is contained in $\mathcal{S}_2(\Phi)$ as well).

Note that $\cS^\C_2(\Phi)$ is a subspace of $P^\C_{2n}$ and is 
spanned by $p_{i} p_{j}$ for $0\le i \le j \le 2n$. 
Let $\tilde \pi$ be the representation of $\g^{\C}$ on $P^\C_{2n}$, i.e., 
$$
\tilde \pi: (f, \tilde \phi)\in \g^\C \times P^\C_{2n} \mapsto \tilde \pi(f) \tilde \phi:=f \tilde \phi \in P^\C_{2n}. 
$$
We have the following fact: 

\begin{lemma}\label{lem:invariancecs2}
The subspace $\cS^\C_2(\Phi)$ is invariant under $\tilde \pi(\g)$.  	
\end{lemma}
   
\begin{proof}{Proof.}
Because $\cS^\C_2(\Phi)$ is spanned by $p_{i}p_j$, for $0\le i \le j \le 2n$, 
	it suffices to show that for any such $p_ip_j$ and for any $f\in \g^\C$,  $\tilde \pi(f)(p_ip_j)$ belongs to $\cS^\C_2(\Phi)$. 
But, this directly follows from	the Leibniz rule, 
$$
\tilde \pi(f) (p_ip_j) = (fp_i)p_j + p_i(fp_j). 
$$
Note that both $fp_i$ and $fp_j$ belong to $H^\C_n$ because $H^\C_n$ is invariant under $\pi(\g)$. Thus, the right hand side of the above expression belongs to $\cS^\C_2(\Phi)$.   
\end{proof}

By Lemma~\ref{lem:invariancecs2}, one can obtain a representation of $\g^\C$ on $\cS^\C_2(\Phi)$ by restricting $\tilde \pi$ to $\g^\C\times \cS^\C_2(\Phi)$. With slight abuse of notation, we will still use $\tilde \pi$ to denote such a representation.  
The representation $\tilde \pi$ is, in general, not irreducible. But, by Lemma~\ref{prop:harmonicpolynomial}, we know that there exist a positive integer $N$ and nonnegative integers $0\le k_1 < \cdots < k_N \le n$ such that 
$$
\cS^\C_2(\Phi) = \|x\|^{2k_1} H^\C_{2n - 2k_1} \oplus \cdots \oplus \|x\|^{2k_N} H^\C_{2n - 2k_N}.
$$
Moreover, $\tilde \pi$ is an irreducible representation when restricted to every subspace $\|x\|^{2k_i} H^\C_{2n - 2k_i}$ for $i = 1,\ldots, N$.

Note, in particular, that if $k_N = n$, then $H^\C_0 = \C$ and, hence,  $\cS^\C_2(\Phi)$ contains the desired polynomial $\|x\|^{2n}$. We show below that this is indeed the case:

\begin{proof}{Proof of Prop.~\ref{prop:qqq}.}
Consider the following element in $\cS^\C_2(\Phi)$: 
$$q^*:= \sum^{2n}_{k = 0} (-1)^{n + k} p_k p_{2n - k}.$$ We show below that 
$q^* =  c\|x\|^{2n}$ 
for some $c > 0$. 
First, note that by~\eqref{eq:hshshs}, $q^*$ can be re-written as follows:
$$
q^* =  q^2_n + 2\sum^{n-1}_{k = 0} \frac{(2n - k)!}{k!} |p_k|^2. 
$$
Note that $q_n$ is real, so $q^*$ is strictly positive.   

We next show that both $\tilde \pi(h)q^*$ and $\tilde \pi(e_+)q^*$ are zero. 
For $\tilde \pi(h)q^*$, we have that 
	$$
	\tilde \pi(h)q^* =  \sum^{2n}_{k = 0} (-1)^{n + k}
	\tilde \pi(h)(p_k p_{2n - k}) = \sum^{2n}_{k = 0} (-1)^{n + k}\left ( (h p_k) p_{2n - k} + p_k (hp_{2n - k}) \right ) = 0,  
	$$
	where the last equality follows from~\eqref{eq:relationship1}. 
 
For $\tilde \pi(e_+)q^*$, we use the fact that $e_+ p_0 = 0$ (because $p_0$ is a highest weight vector) and obtain that  
\begin{equation}\label{eq:asdasd}
\tilde \pi(e_+)q^* = \sum^{2n - 1}_{k = 0} (-1)^{n + k} (p_k (e_+ p_{2n - k})  - (e_+ p_{k+1}) p_{2n - k - 1} )
\end{equation}
It follows from~\eqref{eq:reverse} that
$$
\left\{
\begin{array}{lll}
	e_+p_{2n - k} & = & (2n - k)(k + 1) p_{2n - k -1}, \\ 
	e_+p_{k + 1} & = & (k + 1)(2n - k) p_{n - 2k}, 
\end{array}
\right. 
$$ 
and, hence, each addend on the right hand side of~\eqref{eq:asdasd} is~$0$. 

We now let $\C q^*$ be the one-dimensional subspace of $\cS^\C_2(\Phi)$ spanned by~$q^*$.  Because both $\tilde \pi(h)q^*$ and $\tilde \pi(e_+)q^*$ are zero,  
we obtain by Lemma~\ref{lem:highestweight} that $\tilde \pi$ is an irreducible representation when restricted to $\g^\C \times \C q^*$. Moreover, its highest weight of the representation is~$0$. Thus, by Lemma~\ref{prop:harmonicpolynomial},  
$$\C q^* = \|x\|^{2n} H^\C_0.$$ 
Since $q^*$ is positive, we conclude that $q^* = c \|x\|^{2n}$ for some positive constant $c$.  
\end{proof}

\subsection{Proof of Theorem~\ref{thm:mthm}}\label{ssec:proofmain}
In the section, we prove Theorem~\ref{thm:mthm}. Besides the results established in the previous subsections, we also need the following fact:

\begin{lemma}\label{lem:hhh}
For two points~$x$ and~$x'$ in~$S^2$, if $p(x) = p(x')$ for all $p\in H_n$ with $n$ positive, then $x' \in \{x, (-1)^{n-1} x\}$.  	
\end{lemma}

A proof of the lemma is provided in Appendix-E. 
We are now in a position to prove Theorem~\ref{thm:mthm}:

\begin{proof}{Proof of Theorem~\ref{thm:mthm}.}
Let $(x_\Sigma(0),\mu)$ be an arbitrary pair and $(x'_\Sigma(0),\mu')$ be any pair that is output equivalent to $(x_\Sigma(0),\mu)$. We show below that $(x'_\Sigma(0),\mu')$ is either $(x_\Sigma(0),\mu)$ or $((-1)^{n-1}x_\Sigma(0), \mu)$.

Let $\Phi:= \{p_i\}^{2n + 1}_{i = 1}$ be an arbitrary basis of $H_n$.  Then, by Prop.~\ref{prop:ppp}, we obtain that for any $\sigma\in \Sigma$ and any $i = 1,\ldots, 2n + 1$, 
\begin{equation}\label{eq:item1}
p_i(x_\sigma(0))\rho(\sigma) = p_i(x'_\sigma(0))\rho'(\sigma).  
\end{equation}
Next, by Prop.~\ref{prop:qqq}, there exists a quadratic form~$q$ in~$p_i$ such that the following holds:   
$$q(x) = \sum_{1\le i \le j \le 2n + 1} c_{ij} p_i(x)p_j(x) = 1, \quad \forall x\in S^2.$$ 
It then follows from~\eqref{eq:item1} that for all $\sigma\in \Sigma$,  
$$
\rho^2(\sigma) = \rho^2(\sigma) q(x_\sigma(0))= \rho'^2(\sigma) q(x'_\sigma(0)) =  \rho'^2(\sigma).   
$$
Because the two density functions $\rho$ and $\rho'$ are nonnegative everywhere, we obtain that $$\rho(\sigma) = \rho'(\sigma), \quad \forall \sigma\in \Sigma.$$

Furthermore, it follows from~\eqref{eq:item1} that for any $\sigma\in \Sigma$ and any $i = 1,\ldots, 2n + 1$, 
\begin{equation*}
p_i(x_\sigma(0)) = p_i(x'_\sigma(0)). 
\end{equation*}
Because $\{p_i\}^{2n + 1}_{i = 1}$ form a basis of $H_n$, we have that $p(x_\sigma(0)) = p(x'_\sigma(0))$ for all $p\in H_n$ and for all $\sigma\in \Sigma$.  
Thus, by Lemma~\ref{lem:hhh}, we obtain that
\begin{equation}\label{eq:inno}
x_\sigma(0) \in \{ x_\sigma(0), (-1)^{n - 1} x_\sigma(0) \}, \quad \forall \sigma\in \Sigma.
\end{equation}
Note, in particular, that if $n$ is odd, then $x'_\sigma(0) = x_\sigma(0)$  for all $\sigma\in \Sigma$. Thus, in this case, system~\eqref{eq:system} is ensemble observable. We now assume that~$n$ is even and show that $x'_\Sigma(0)$ is either $x_\Sigma(0)$ or $-x_\Sigma(0)$. But, this follows from the fact that both $x_\sigma(0)$ and $x'_\sigma(0)$ are continuous in~$\sigma$. To see this, consider a map $\delta: \Sigma\to \R$ defined by sending $\sigma$ to the Euclidean distance between $x_\sigma(0)$ and $x'_\sigma(0)$, i.e.,  
$$
\delta: \sigma \mapsto \delta(\sigma):=\|x_\sigma(0) - x'_\sigma(0)\|.
$$     
Because $x_\sigma(0)$ and $x'_\sigma(0)$ are continuous in $\sigma$, the map~$\delta$ is continuous as well. On the other hand, we note that by~\eqref{eq:inno}, there are only two cases: 
\begin{enumerate}
\item[\em (1)] If $x'_\sigma(0) = x_\sigma(0)$, then $\delta(\sigma) = 0$. 
\item[\em (2)] If $x'_\sigma(0) = -x_\sigma(0)$, then $\delta(\sigma) = 2$. 
\end{enumerate}
Thus, if $x'_\sigma(0) = x_\sigma(0)$ (resp. $x'_\sigma(0) = -x_\sigma(0)$) for a certain $\sigma\in \Sigma$, then by continuity of~$\delta$, 
$x'_\Sigma(0) = x_\Sigma(0)$ (resp. $x'_\Sigma(0) = -x_\Sigma(0)$). This completes the proof.     
\end{proof}

\section{Conclusions}\label{sec:conclusion}
We have addressed in the paper the problem about observability of a continuum ensemble of Bloch equations~\eqref{eq:system}. We assume that the initial states $x_\sigma(0)$ of the individual systems are unknown and, moreover, the measure~$\mu$ that describes the overall population density of the individual systems is also unknown. The problem is about whether one is able to estimate $x_\sigma(0)$ for every $\sigma\in \Sigma$ and the measure~$\mu$ using only a scalar measurement output~$y(t)$.

We have provided a class of observation functions $\phi$ that guarantee (weak) ensemble observability of the resulting system~\eqref{eq:system}. Specifically, we have shown that if $\phi$ is a harmonic homogeneous polynomial of positive degree, then two pairs $(x_\Sigma(0), \mu)$ and $(x'_\Sigma(0), \mu')$ are output equivalent if and only if $\mu = \mu'$ and $x'_\Sigma(0) \in \{x_\Sigma(0), (-1)^{n - 1} x_\Sigma(0)\}$. %In particular, if $n$ is odd, then system~\eqref{eq:system} is ensemble observable. 
 
The proof of the result relies on the use of representation theory of~$\sl(2,\C)$. In particular, the following two items are key to establishing the result:  
\begin{enumerate} 
	\item[\em (1)] We have introduced the Casimir element $\eta^*$ (and its variants $\xi$ and $\zeta$ defined in Lemma~\ref{lem:defxizeta}) which acts  on  the space of harmonic homogeneous polynomials as a scalar multiple of the identity operator. This fact is key to establishing Prop.~\ref{prop:ppp}. 
	\item[\em (2)] We have used the fact that any finite-dimensional representation of $\sl(2,\C)$ is reducible and, then, decomposed the space of quadratic forms $\cS^\C_2(\Phi)$ (with $\Phi$ a basis of $H^\C_n$) into a direct sum of invariant subspaces under the representation. In particular, we have shown that $\cS_2(\Phi)$ contains the one-dimensional subspace spanned by $\|x\|^{2n}$, which is the constant function~$\mathbf{1}_{S^2}$ on $S^2$. This fact is key to establishing Prop.~\ref{prop:qqq}.     
\end{enumerate}
 The approach developed in the paper can be extended to analyze observability of other ensemble systems defined on Lie groups and their homogenous spaces. The above two items could serve as guidelines for the extension.

\bibliographystyle{IEEEtran}    
\bibliography{estimationbloch.bib}

\appendix

\section{Proof of Lemma~\ref{lem:casimir}}
It suffices to show that $\eta^*$ commutes with every~$f_i$ for $i = 0,1,2$.
	Recall that if $(i,j,k)$ is a cyclic rotation of $(0,1,2)$, then $[f_i,f_j] = f_k$. Thus, by symmetry, we only need to show that~$\eta^*$ commutes with~$f_0$. First, note that
	\begin{multline*}
	f_0f^2_1 = f_1^2f_0  + (f_0f_1 - f_1f_0)f_1 + f_1(f_0f_1 -f_1f_0) \\
	= f_1^2f_0  + [f_0,f_1]f_1 + f_1[f_0,f_1] = f^2_1f_0 + f_2f_1 + f_1f_2.
	\end{multline*}
	Similarly, we obtain that
	\begin{multline*}
	f_0f^2_2 = f_2^2f_0  + (f_0f_2 - f_2f_0)f_2 + f_2(f_0f_2 -f_2f_0) \\
	= f_2^2f_0  + [f_0,f_2]f_2 + f_2[f_0,f_2] = f^2_2f_0 - f_1f_2 - f_2f_1.
	\end{multline*}
	It then follows that $f_0$ commutes with $(f^2_1 + f^2_2)$ and, hence, with~$\eta^*= \sum^2_{i = 0} f^2_i$ as well. 
\hfill{\qed}

\section{Proof of Lemma~\ref{prop:harmonicpolynomial}}

   Let $h$, $e_+$, $e_-$ be defined in~\eqref{eq:triplet}. Then, by computation, we obtain that
     $$
     \pi(h) p^*_k = 2(n - 2k) p^*_k  \quad \mbox{and} \quad \pi(e_+) p^*_k = 0. 
     $$
     Let $V_k$ be a subspace of~$P_n^\C$ spanned by $\pi^l(e_-)p^*_k$ for $l = 0, \ldots, 2(2n-k)$. Then, by Lemmas~\ref{lem:irredrep} and~\ref{lem:highestweight}, 
     it suffices to show that $V_k = \|x\|^{2k}\BH^\C_{n - 2k}$. 
     
     First, note that the dimension of $H^\C_{n - 2k}$ is $2(n - 2k) + 1$, which is the same as the dimension of the subspace~$V_k$. Thus, we only need to show that each $\pi^{l}(e_-)p^*_k$, for $l = 0,\ldots, 2(n - 2k)$, belongs to $\|x\|^{2k}H^\C_{n - 2k}$.  
     
     Next, note that $\|x\|^2 = 1$ for all $x\in S^2$. Thus, for any $i = 1,2,3$,  $f_i \|x\|^2 = 0$ and, hence, $\pi(e_-) \|x\|^2 = 0$. It follows that for any $l = 0,\ldots, 2(n - 2k)$,  
     $$\pi^{l}(e_-)p^*_k  = \|x\|^{2k} \pi^{l}(e_-)(x_1 + \mathrm{i} x_2)^{n - 2k}.     
     	$$      
     
     It now remains to show that each $\pi^{l}(e_-)(x_1 + \mathrm{i} x_2)^{n - 2k}$, for $l = 0,\ldots, 2(n - 2k)$, belongs to $H^\C_{n - 2k}$.  
    To see this, note that the Laplace operator $\Laplace$ commutes with every $f_i$, i.e., $\Laplace f_i  = f_i \Laplace$ for all $i = 1,2,3$. In particular, it commute with $\pi(e_-)$. Thus, 
     $$
     \Laplace \pi^{l}(e_-) (x_1 + \mathrm{i} x_2)^{n - 2k} = \pi^{l}(e_-) \Laplace (x_1 + \mathrm{i} x_2)^{n - 2k} = 0  
     $$
     for all $l = 0,\ldots, 2(n - 2k)$. 
     \hfill{\qed}

\section{Proof of Lemma~\ref{lem:diagonaloperator}}
Because $\eta^*$ belongs to the center of $U(\g)$, 
	$\eta^* f = f \eta^*$ for all $f\in \g$. Then, by Schur's Lemma, there exists a constant $\lambda\in \C$ such that $\eta^* p = \lambda p$ for all $p\in H^\C_n$. To evaluate $\lambda$, we let $p^* := (x_1 + \mathrm{i}x_2)^n$ be a highest weight vector in $H^\C_{n}$ (with the highest weight being~$2n$). Next, let $h$, $e_+$, and $e_-$ be defined in~\eqref{eq:triplet}. Note that
	$$
	\eta^* = \sum^2_{i = 0}f^2_i = -\frac{1}{4}h^2 - \frac{1}{2}(e_+e_- + e_-e_+).
	$$ 
	Then, using the fact that $e_+ p^* = 0$, we obtain that
	$$
	\eta^*p^* = -\frac{1}{4}h^2 p^* - \frac{1}{2}(e_+e_- + e_-e_+) p^*  
	= -\frac{1}{4}h^2 p^*  - \frac{1}{2}(e_+e_- - e_-e_+)p^*. 
	$$
	Further, note that $[e_+, e_-] = h$ and $h p^* = 2n p^*$. Thus, 
	$$
	\eta^*p^* = -\left (\frac{1}{4}h^2 + \frac{1}{2} h\right ) p^* = -n(n + 1) p^*, 
	$$     
	which implies that $\lambda = -n(n + 1)$.   	    
\hfill{\qed}

\section{Proof of Lemma~\ref{lem:niu}}
We first show that for each $k = 0,\ldots, n$, there exists a complex number $c_k$ such that 
$p_{2n - k} = c_k \bar p_k$. To see this, we note that  
$h p_k = (2n - 2k) p_k$. Taking complex conjugate on both sides, we have $\bar h \bar p_k = (2n - 2k) \bar p_k $. Recall that $h = 2\mathrm{i} f_0$, so $\bar h = - h$. Thus, 
$
h \bar p_k = -(2n - 2k) \bar p_k
$, 
so $\bar p_k$ belongs to the weight space corresponding to the weight $-(2n - 2k)$. Because the weight space is one-dimensional (over $\C$) and because $p_{2n - k}$ belongs to the same weight space, there exists a $c_k\in \C$ such that $p_{2n - k} = c_k \bar p_k$. 

Next, we note that $c_k \bar e_+ \bar p_k = \bar e_+ p_{2n -k}$. 
Recall that by~\eqref{eq:reverse}, $e_+ p_k = k(2n - k + 1) p_{k - 1}$, so $\bar e_+ \bar p_k = k(2n - k + 1) \bar p_{k - 1}$.  
From~\eqref{eq:triplet}, we have that $\bar e_+ = -e_-$ and, hence, $\bar e_+ p_{2n -k}= - e_- p_{2n - k} = - p_{2n - k + 1}$. It then follows that
$$
c_k k(2n - k + 1) \bar p_{k - 1}  = c_k \bar e_+ \bar p_k = \bar e_+ p_{2n -k}  = - p_{2n - k + 1} = c_{k - 1}\bar p_{k - 1},
$$
which then implies that
\begin{equation*}
c_{k - 1} = -k(2n - k + 1) c_k.  	
\end{equation*}
As a consequence, the following holds:
$$
c_k = c_n (-1)^{n-k} \frac{(2n - k)!}{k!}, \quad \forall k = 0,\ldots, n.
$$ 

To establish the lemma, it now suffices to show that $c_n = 1$. Note that $c_n$ satisfies the condition $c_n \bar p_n = p_n$. As a consequence, $c_n = 1$ if and only if $p_n$ is real. 
We write $p_n = \sum^l_{i = 1} \gamma_i m_i$, where $m_i$ are monomials in variables $x_1$, $x_2$, and $x_3$ and $\gamma_i\in \C$ are coefficients.  
Note that if there is some $i = 1,\ldots,l$, such that $\gamma_i$ is real, then all the coefficients are real. This holds because otherwise, $p_n$ and $\bar p_n$ are linearly independent which contradicts the fact that they both belong to the same weight space. With that in mind, we show below that $p_n$ contains the monomial $x^n_3$ with nonzero, real coefficient. Recall that $p_n = e^n_- p_0$ where $p_0 = (x_1 + \mathrm{i} x_2)^n$ and $e_- = -f_1 + \mathrm{i} f_2$ with $f_1$ and $f_2$ defined in~\eqref{eq:vectorfields}.  Straightforward computation shows that the coefficient of $x^n_3$ in $p_n$ is given by $n! (-2)^n $. 
\hfill{\qed}

\section{Proof of Lemma~\ref{lem:hhh}}
Recall that $H_n^\C$ is the complexification of $H_n$. 
We fix an arbitrary $x\in S^2$ and show that if $p(x') = p(x)$ for all $p \in H^\C_n$, then $x' \in \{x, (-1)^{n - 1} x\}$. Since the $x_i$'s cannot be zero simultaneously, we assume without loss of generality that $x_3 \neq 0$.  
Then, consider the following three homogeneous polynomials in $H^\C_n$:  
$$
p_1(x) := (x_1 + \mathrm{i} x_2)^n, \quad
p_2(x) :=  x_3(x_1 + \mathrm{i} x_2)^{n-1}, \quad
p_3(x):= (x_3 + \mathrm{i}  x_1)^n.
$$
We assume that the values of the above polynomials at the given~$x$ are given by $$p_1(x) = c_1, \quad p_2(x) = c_2, \quad p_3(x) = c_3,$$ for some $c_1, c_2, c_3\in \C$. We provide below solutions~$x'$ to the above polynomial equations. 

If both $x_1$ and $x_2$ are $0$, then, $c_1 = c_2 = 0$ and $c_3$ is a (nonzero) real number. It follows that $x'_1 = x'_2 = 0$ and $x'^n_3 = x^n_3 = c_3$. Thus, in this case, $x' \in \{x, (-1)^{n - 1}x\}$.  
Next, we assume that $x_1^2 + x_2^2 \neq 0$. Since $x_3\neq 0$, every $c_i$ is nonzero. Then,    
$$
\frac{p_1(x')}{p_2(x')} =  \frac{x'_1 + \mathrm{i}  x'_2}{x'_3} = \frac{c_1}{c_2}. 
$$ 		
Since $x'_1$, $x'_2$, and $x'_3$ are real, we have that
\begin{equation}\label{eq:sssss}
x'_1 = \operatorname{re}(\nicefrac{c_1}{c_2})x'_3 \quad \mbox{and} \quad x'_2 = \operatorname{im}(\nicefrac{c_1}{c_2})x'_3. 
\end{equation}
where $\operatorname{re}(\cdot)$ and $\operatorname{im}(\cdot)$ denote the real  and imaginary part of a complex number, respectively. On the other hand, we also have that $\sum^3_{i = 1} x'^2_i = 1$. Thus,~\eqref{eq:sssss} determines $x'$ up to sign, i.e., $x' = \pm x$. If, further, $n$ is odd, then $$p_1(-x) = -p_1(x) = -c_1 \neq c_1,$$ and, hence, $x'$ can only be~$x$. Combining the above arguments, we conclude that $x' \in \{x, (-1)^{n - 1}x\}$.  
\hfill{\qed}

\section{The Addition Theorem}
We provide here another proof of Prop.~\ref{prop:qqq} using the Addition Theorem for spherical harmonics (see, for example,~\cite[Ch.~12]{arfken2005mathematical}). Recall that the Cartesian coordinate system $(x_1,x_2,x_3)$ and the spherical coordinate system $(r, \theta, \varphi)$ are related by
\begin{equation}\label{eq:sphericalcoord}
x_1 = r\sin\theta \cos \varphi, \quad
x_2 = r\sin \theta\sin \varphi, \quad
x_3 = r\cos\theta.
\end{equation} 
We next recall that {\em spherical harmonics} $Y^k_n(\theta, \varphi)$ are defined as follows: For a given a nonnegative integer $n$ and an integer $k$ with $|k| \le n$, we have that
$$
Y^k_{n}(\theta, \varphi) := (-1)^k \sqrt{\frac{2n + 1}{4\pi} \frac{(n - k)!}{(n + k)!}} L^k_n(\cos\theta) e^{\mathrm{i} k\varphi},
$$  
where $L^k_{n}$ is the {\em associated Legendre polynomial} define by
$$
L_n^k(x):= \frac{(-1)^k}{2^n n!} (1 - x^2)^{k/2} \frac{d^{n + k}}{dx^{n + k}} (x^2 - 1)^n. 
$$
It is known that $\{Y^k_{n}\}^n_{k = -n}$ is a basis of $H^\C_n$ (after change of coordinates~\eqref{eq:sphericalcoord}). In other words, each harmonic homogeneous polynomial $p\in H^\C_n$ can be expressed as a linear combination of the spherical harmonics and vice versa. In fact, we note here that each $Y^k_n$ for $|k| \le n$ is linearly proportional to $p_{n - k}$ where $p_{n - k}$ is defined in~\eqref{eq:defphikkk}.

We now reproduce the Additional Theorem for spherical harmonics: First, recall that the ordinary Legendre polynomial $L_n$ can be described by the Rodrigues' formula:  
$$
L_n(x) := \frac{1}{2^n n!} \frac{d^n}{dx^n} (x^2 - 1)^n.
$$
Next, for two points $(1, \theta, \varphi)$ and $(1, \theta', \varphi')$ on the unit sphere $S^2$, we let $\gamma$ be the angle between these two points, i.e., 
$$
\cos \gamma = \cos\theta\cos\theta' + \sin\theta\sin\theta' \cos(\varphi - \varphi').
$$ 
Then, the Addition Theorem for spherical harmonics is given by the following:   

\begin{lemma}[Addition Theorem]
For any two pairs $(\theta, \varphi)$ and $(\theta', \varphi')$, we have that
$$
L_n(\cos \gamma) := \frac{4\pi}{2n + 1} \sum^n_{k = -n} Y_n^k(\theta, \varphi) \bar Y^{k}_n(\theta',\varphi'),
$$ 
where $\bar Y^{k}_n(\theta',\varphi')$ is the complex conjugate of $Y^{k}_n(\theta',\varphi')$.
\end{lemma}
  
Prop.~\ref{prop:qqq} is then a corollary to the above result. To see this, we let $(\theta, \varphi) = (\theta',\varphi')$. Then, by the Addition Theorem, we have that for any $(\theta,\varphi)$, 
$$
 \frac{4\pi}{2n + 1} \sum^n_{k = -n} |Y_n^k(\theta, \varphi)|^2 = L_n(1). 
$$ 
%The above equality is also known as the Uns{\"o}ld's theorem~\cite{}. 
Finally, note that $L_n(1) = 1$ for any $n \ge 1$, which then completes the proof of Prop.~\ref{prop:qqq}.   \hfill{\qed}
\end{document}